\documentclass[sn-mathphys-num]{sn-jnl}

\usepackage{graphicx}%
\usepackage{multirow}%
\usepackage{amsmath,amssymb,amsfonts}%
\usepackage{amsthm}%
\usepackage{mathrsfs}%
\usepackage[title]{appendix}%
\usepackage{xcolor}%
\usepackage{textcomp}%
\usepackage{manyfoot}%
\usepackage{booktabs}%
\usepackage{algorithm}%
\usepackage{algorithmicx}%
\usepackage{algpseudocode}%
\usepackage{listings}%

\theoremstyle{thmstyleone}%
\newtheorem{theorem}{Theorem}
\theoremstyle{thmstyletwo}%
\newtheorem{example}{Example}%
\newtheorem{remark}{Remark}%

\newtheorem{corollary}{Corollary}
\newtheorem{construction}{Construction}

\theoremstyle{thmstylethree}%
\newtheorem{definition}{Definition}%

\begin{document}

\title[Multiple Spectrally Null Constrained Complete Complementary Codes of Various Lengths Over Small Alphabet]{Multiple Spectrally Null Constrained Complete Complementary Codes of Various Lengths Over Small Alphabet}


\author*[1]{\fnm{Rajen} \sur{Kumar} \href{https://orcid.org/0000-0003-2100-6447}{\includegraphics[scale=0.01]{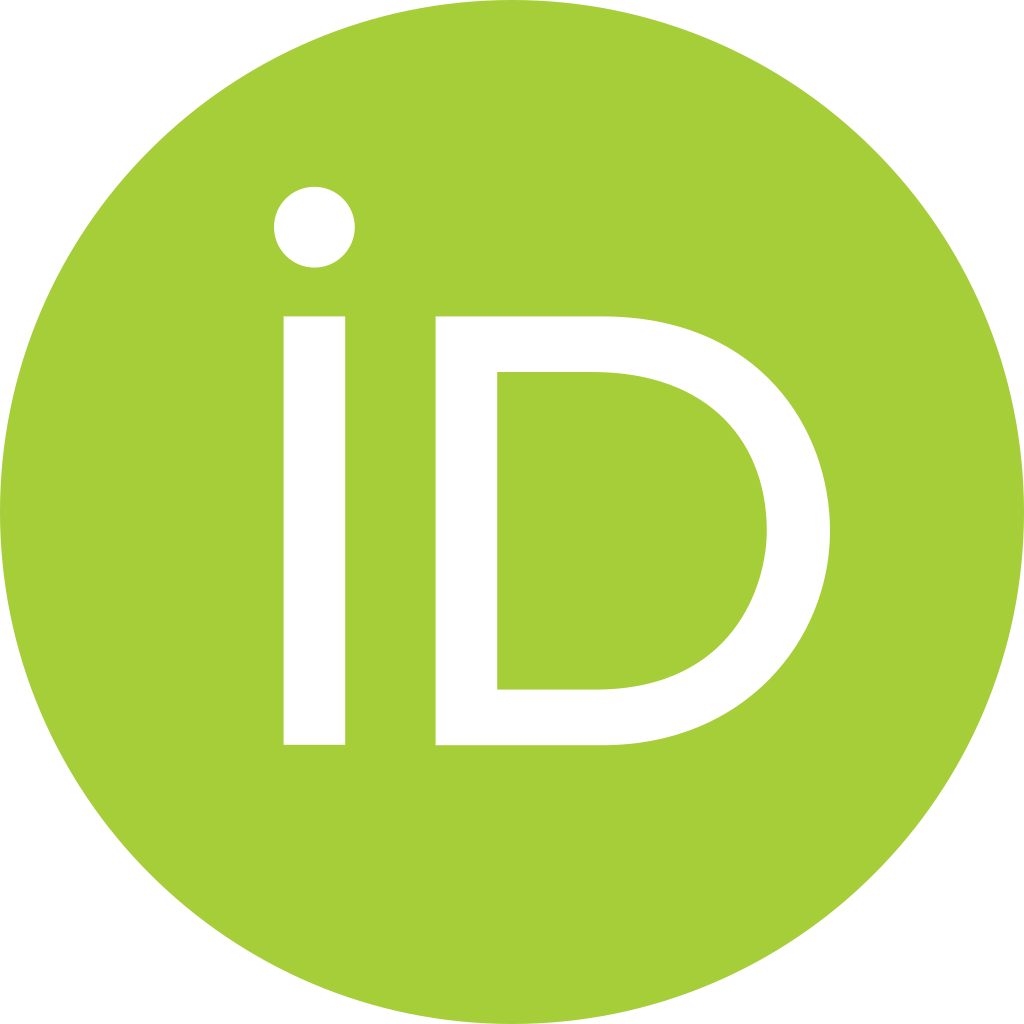}}}\email{rajen\_2021ma4@iitp.ac.in}

\author[2]{\fnm{Palash} \sur{Sarkar} \href{https://orcid.org}{\includegraphics[scale=0.01]{orcid_logo.png}}}\email{palash.sarkar@uib.no}

\author[1]{\fnm{Prashant Kumar} \sur{Srivastava} \href{https://orcid.org/0000-0002-7651-5639}{\includegraphics[scale=0.01]{orcid_logo.png}}}\email{pksri@iitp.ac.in}


\author[3]{\fnm{Sudhan} \sur{Majhi} \href{https://orcid.org/0000-0002-2142-1862}{\includegraphics[scale=0.01]{orcid_logo.png}}}\email{smajhi@iisc.ac.in}

\affil[1]{\orgdiv{Department of Mathematics}, \orgname{Indian Institute of Technology}, \orgaddress{\city{Patna}, \postcode{801106}, \state{Bihar}, \country{India}}}

\affil[2]{\orgdiv{Department of Informatics}, \orgname{University of Bergen}, \orgaddress{\city{Selmer Center}, \country{Norway}}}

\affil[3]{\orgdiv{Department of Electrical Communication Engineering}, \orgname{Indian Institute of Science}, \orgaddress{ \city{Bangalore}, \postcode{560012}, \state{Karnataka}, \country{India}}}


\abstract{Complete complementary codes (CCCs) are highly valuable in the fields of information security, radar and communication. The spectrally null constrained (SNC) problem arises in radar and modern communication systems due to the reservation or prohibition of specific spectrums from transmission. The literature on SNC-CCCs is somewhat limited in comparison to the literature on traditional CCCs. The main objective of this paper is to discover several configurations of SNC-CCCs that possess more flexibility in their parameters. The proposed construction utilises the existing CCCs and mutually orthogonal sequences and covers all lengths with the smallest alphabets $\{-1,0,1\}$. Further, SNC-CCC is extended to multiple SNC-CCCs with an inter-set zero cross-correlation zone (ZCCZ). We could control the cross-correlation magnitude outside the ZCCZ, through the proposed construction. Consequently, the resulting codes possess both aperiodic and periodic inter-set ZCCZ and feature a low magnitude of cross-correlation value outside the ZCCZ.}




\maketitle

A Golay complementary pair (GCP) indicates a pair of sequences whose sum of aperiodic auto-correlation functions (AACFs) results in zero at non-zero time shifts. Golay uncovered a sequence pair that can be used during the research of multislit spectroscopy \cite{Golay:51, Golay_1961}. GCPs are comprehensively utilised in engineering applications, particularly in radar systems and communication systems. These applications include channel estimation \cite{Spasojevic, Wang_IET}, design of the physical uplink control channel \cite{Sahin}, non-orthogonal multiple access \cite{Yu}, radar waveform design \cite{Pezeshki}, control of peak-to-average power ratio for multi-carrier communication systems \cite{Davis}, and more. Tseng and Liu presented the idea of a Golay complementary set (GCS) comprising more than two sequences. The sum of the AACFs for all sequences is zero except at zero time shift \cite{Tseng}. Due to their similar characteristics to GCPs, GCSs are also utilised in several communication and radar systems \cite{Nguyen, Aparicio}. Moreover, GCS has the added benefit of a greater code rate compared to GCP, in addition to its flexible length advantage \cite{Paterson, Schmidt, Sarkar_2020}. 

A mutually orthogonal Golay complementary set (MOGCS) is a collection of $K$ GCSs. Each GCS in the MOGCS has $M$ sequences, each of length $L$. Additionally, the cross-correlation function between distinct GCSs is zero. A MOGCS is referred to as a complete complementary code (CCC) when $K$ is equal to $M$ \cite{Suehiro}. For implementing multi-antenna or multi-user systems, it is important to consider the cross-correlation characteristics across sets of sequences. This is particularly relevant for systems such as CCC based code division multiple access (CDMA) and multi-input multi-output (MIMO) radar \cite{Liu_2021,Sun_2015_Survey,Tang_2014,Sun_2015_VT}. The idea of CCC extended to multiple CCC with an inter-set zero cross-correlation zone (ZCCZ) \cite{Men_MCCC, Tian_MCCC}, which is similar to the $Z$-complementary code set (ZCCS). The idea of CCC is also extended to multiple CCC with inter-set low cross-correlation, which is similar to the Quasi complementary code set (QCCS).

In systems that use orthogonal frequency division multiplexing (OFDM), some sub-carriers are designated as reserved and are not allowed to transmit signals \cite{IEEE_std}. For instance, the direct current sub-carrier is specifically allocated, known as spectrally null constrained (SNC), to prevent any discrepancies in the D/A and A/D converters during radio frequency transmission \cite{Zhou_2020}. The increasing need for OFDM or multi-carrier CDMA sequences with spectrum null constraints, also known as non-contiguous sequences, is primarily motivated by their potential applications in cognitive radio (CR) communications \cite{Hamilton}. Transmission on sub-carriers not used by primary users constrains secondary users in OFDM based CR transmissions. The Third-Generation Partnership Project Long-Term Evolution enhanced licenced-assisted access and the New Radio in Unlicenced (NR-U) implemented interlaced transmission, with the null locations of the SNC sequences being regularly distributed (although the nulls in NR-U are unevenly spaced). It is also important to think about the SNC when using the CCC as omnidirectional precoding for a rectangular array that is not all the same size \cite{Su_2019}. In the IEEE P802.15.4z standard, the average power permitted in ultra-wide-band is very low. Therefore, the sequence design will consider the inclusion of null to decrease the average power. To summarise, several situations in sequence design require the use of null constraints.

Only a few of the conventional GCSs and CCCs take into account this limitation, which has been addressed in \cite{Gavish_ternary, Sahin, Zhou_2020, Sahin_2021, Ipanov, Shen_CCC, Nishant_SNC}. Sahin and Yang extended the conventional GCPs to address the SNC problem, as described in \cite{Sahin} and \cite{Sahin_2021}. In \cite{Zhou_2020}, Zhou \textit{et al.} sequentially built the SNC-MOGCSs/SNC-GCSs using an iterative approach. They used two sequences extracted from a GCP as the initial seed sequence and then introduced a certain number of zeros into these two sequences. As a result, new sequences were obtained with a zero correlation zone. Hence, a challenging issue arises regarding the methodology for constructing SNC-CCC. Shen \textit{et~al.} proposed a method for constructing SNC-CCC using extended Boolean functions and graphs \cite{Shen_CCC}. However, the parameters are only in the power of $p$ when elements of code are considered from the $q$th root of unity and zero, for $p\mid q$ and $p\ge 2$. In machine-type communication, alphabet size plays a major role and must be minimum \cite{Sarkar_QCCS_TIT}. Notably, there are gaps in the SNC-CCC proposed in \cite{Shen_CCC} in terms of lengths and alphabet sizes. For example, when the alphabets are $-1,1,0$, the set size, the code size and the length are restricted to the form of a power-of-two. We are strongly motivated to include a bigger range of parameters for SNC-CCC in comparison to existing literature. The proposed construction not only provides SNC-CCCs with new parameters but also provides flexibility in the alphabet and the length of the constituent sequences. In this paper, codes are referred to as CCC, when it is a traditional CCC, and codes with nulls are referred to as SNC-CCC.

In the proposed construction, we use existing CCCs and mutually orthogonal sequences (MOSs) as seeds. By performing the concatenation operation in a specific way, as described in Section \ref{Sec:Propos}, we obtain multiple SNC-CCCs over a small alphabet. It may be noted that our smallest alphabet is $\{-1,0,1\}$, on which the proposed construction is capable of generating almost all possible lengths. The proposed multiple SNC-CCCs also have a ZCCZ property with respect to both periodic and aperiodic correlation. With these properties, the obtained code set is useful for multi-cell MC-CDMA systems, where the users inside a cell enjoy interference-free communication due to the ideal correlation property of a SNC-CCC and the users from two different cells also enjoy interference-free communication within the ZCCZ. Our study also reveals that we can control non-zero inter-set cross-correlation magnitude values outside the ZCCZ. We consider this an opportunity to minimize the upper bound for inter-set cross-correlation magnitude values of the proposed multiple SNC-CCCs.

We structure the subsequent sections of the paper as follows: Section \ref{Sec:prelim} establishes appropriate notations and definitions. Section \ref{Sec:Propos} introduces new constructions for SNC-CCC and multiple SNC-CCCs and provides an example to illustrate this. Further, we explain the ZCCZ width of the multiple SNC-CCCs and conclude with the low inter-set cross-correlation value. In Section \ref{Sec:comparison}, a comparison has been given with existing literature. Based on the proposed work, we have highlighted three problems that we may consider as our future work in Section \ref{Sec:fut_work}. The paper is concluded in Section \ref{Sec:Con}.
\section{Preliminaries}\label{Sec:prelim}
First, we specify the notation and definitions which are used consistently throughout this paper.

\begin{definition}
Let $\mathbf{a}=(a_1,a_2,\ldots,a_{L})$ and $\mathbf{b}=(b_1,b_2,\ldots,b_{L})$ be
two complex-valued sequences of length $L$ and $\tau$ be an integer. Define
\begin{equation}
    \mathcal{C}(\mathbf{a}, \mathbf{b})(\tau)=\left\{\begin{array}{ll}
		\sum_{i=1}^{L-\tau} a_{i+\tau} b_{i}^{*}, & 0 \leq \tau<L, \\
		\sum_{i=1}^{L+\tau} a_{i} b_{i-\tau}^{*}, & -L<\tau<0, \\
		0, & \text { otherwise, }
	\end{array}\right.
\end{equation}
which is called ACCF of $\mathbf{a}$ and $\mathbf{b}$ at time shift $\tau$, where $ (\cdot)^*$ represents complex conjugation. When $\mathbf{a}=\mathbf{b}$, $\mathcal{C}(\mathbf{a},\mathbf{b})(\tau)$ is called AACF of $\mathbf{a}$ and is denoted by $\mathcal{C}(\mathbf{a})(\tau)$. Further, periodic cross-correlation function (PCCF) of $\mathbf{a}$ and $\mathbf{b}$ at time shift $\tau$ is defined as
\begin{equation}
    \Theta(\mathbf{a}, \mathbf{b})(\tau)=\mathcal{C}(\mathbf{a}, \mathbf{b})(\tau)+\mathcal{C}(\mathbf{a}, \mathbf{b})(\tau-L).
\end{equation}
\end{definition}

\begin{definition}
	Let $\mathbf{C}=\left\{{C}_{k}: 1\le k \le K \right\}$ be a set of $K$ matrices (codes), each having order $M \times L$. And $C_{k}$ is defined as
\begin{equation}
    C_{k}=\left[\begin{array}{c}
\mathbf{c}_{1}^{k}\\\mathbf{c}_{2}^{k}\\ \vdots\\\mathbf{c}_{M}^{k}
	\end{array}\right]_{M \times L},
\end{equation}
	where $\mathbf{c}_{j}^{k}(1\leq j \leq M ,1 \leq k \leq K)$ is the $j$-th row sequence of $C_{k}$.
 Then ACCF between $C_{k_{1}}$ and $C_{k_{2}}$ is defined by
	\begin{equation}
		\mathcal{C}\left({C}_{k_{1}},{C}_{k_{2}}\right)(\tau)=\sum_{\nu=1}^{M} \mathcal{C}\left(\mathbf{c}_{\nu}^{k_{1}}, \mathbf{c}_{\nu}^{k_{2}}\right)(\tau).
	\end{equation}
When ${C}_{k_1}={C}_{k_2}$, $\mathcal{C}(C_{k_1},C_{k_2})(\tau)$ is called AACF of ${C}_{k_1}$ and is denoted by $\mathcal{C}(C_{k_1})(\tau)$. Similarly, the PCCF of between $C_{k_{1}}$ and $C_{k_{2}}$ is defined by
\begin{equation}
		\Theta\left({C}_{k_{1}},{C}_{k_{2}}\right)(\tau)=\sum_{\nu=1}^{M} \Theta\left(\mathbf{c}_{\nu}^{k_{1}}, \mathbf{c}_{\nu}^{k_{2}}\right)(\tau).
	\end{equation}
\end{definition}
\begin{definition}
    Let $\mathbf{a}=(a_1,a_2,\ldots,a_L)$ be any complex-valued sequence and $N=\{x: a_x = 0, \text{ for }1 \le x \le L\}$ is non-empty set, $\mathbf{a}$ is called a SNC sequence. A CCC is called an SNC-CCC if there is at least one SNC sequence in the CCC \cite{Shen_CCC}.
\end{definition}
\begin{definition}
    Let $\mathbf{C}=\left\{{C}_{k}: 1\le k \le M \right\}$ be a set of $M$ codes of order $M \times L$ and it follows
    \begin{equation}
        \mathcal{C}(C_{k_1},C_{k_2})(\tau)=\left\{ \begin{array}{cc}
           ML-\epsilon,  & k_1=k_2, \tau=0 \\
           0 , & \text{otherwise,}
        \end{array}\right.
    \end{equation}
    where, $\epsilon$ is the number of zeros in a code. When $\epsilon=0$, it is referred to as traditional aperiodic CCC. Note that aperiodic CCC also satisfies the ideal periodic correlation properties.
     
     To avoid possible confusion between the terms aperiodic and periodic CCC, we will exclusively use the term CCC throughout this paper.
\end{definition}
\begin{definition}
    Let $\mathfrak{C}=\{\mathbf{C}^j:1\le j\le P\}$ be a collection of $P$ many $(M,L)$-CCCs, i.e., $\mathbf{C}^j=\left\{{C}^j_{k}: 1\le k \le M \right\}$, where $1\le j \le P, P\ge 2$. If any two codes from different CCCs $\mathbf{C}^{j_1}$ and $\mathbf{C}^{j_2}$ with $1\le j_1\ne j_2 \le P$ follows
    \begin{equation}\label{Eq:cross_interset}
    \begin{aligned}
        \mathcal{C}({C}^{j_1}_{k_1},{C}^{j_2}_{k_2})(\tau)&=0, ~\lvert \tau \rvert <Z_A,\\
        \delta_{A}&=\max \left\{\lvert \mathcal{C}(C^{j_1}_{k_1},C^{j_2}_{k_2})(\tau) \rvert:j_1\ne j_2, 1\le k_1,  k_2 \le M, Z \le \lvert \tau \rvert \le L-1 \right\},
    \end{aligned}  
    \end{equation}
    where $1\le k_1,k_2\le M$, denote $\mathfrak{C}$ is aperiodic $(P,M,L,Z_A,\delta_A)$-CCCs. \\ Similarly, 
      \begin{equation}\label{Eq:per_cross_interset}
    \begin{aligned}
        \Theta({C}^{j_1}_{k_1},{C}^{j_2}_{k_2})(\tau)&=0, ~\lvert \tau \rvert <Z_P,\\
        \delta_{P}&=\max \left\{\lvert \Theta(C^{j_1}_{k_1},C^{j_1}_{k_2})(\tau) \rvert:j_1\ne j_2, 1\le k_1,  k_2 \le M, Z \le \lvert \tau \rvert \le L-1 \right\},
    \end{aligned}  
    \end{equation}
    where $1\le k_1,k_2\le M$, then $\mathfrak{C}$ is  periodic $(P,M,L,Z_P,\delta_P)$-CCCs. 
\end{definition}

Let $\mathbf{a}$ and $\mathbf{b}$ be two complex sequences of identical length and be orthogonal, i.e., their dot product $\langle\mathbf{a},\mathbf{b}\rangle$ is equal to $0$. We refer to the set of sequences as MOSs when the number of sequences exceeds two and the dot product of any two sequences is zero. A construction of $P$ many MOSs with length $P$ is suggested in \cite{Rajen_DZCCS}.

With the help of CCCs, one can easily construct SNC-CCCs by inserting zeros in the following manners.
\begin{remark}\label{Lem:zero_in_CCC}
    Zero insertion properties of CCCs to construct SNC-CCCs are as follows:
\begin{enumerate}
    \item Adding some zero columns at the beginning of each code of CCC,
    \item Adding some zero columns at the last of each code of CCC,
    \item Adding the same number of zero columns between every column of each code of CCC.
\end{enumerate}
\end{remark}
The proof of the above lemma is straightforward.  Therefore, we have omitted the proof.
\section{Proposed Construction}\label{Sec:Propos}
In this section, we describe our main method of construction. First, we provide a new method, which involves the concatenation of zeros and matrices with some scalar multiplications. Scalars must be selected meticulously to ensure they do not impact the elements of matrices following multiplication.
\begin{construction}
Let $C_1,C_2,\ldots,C_P$ be a set of $M\times L$ matrices and $\mathbf{b}=(b_1,b_2,\ldots,b_P)$ be a sequence of length $P$. Define 
\begin{equation}\label{eq:consR}
    \mathcal{R}^{\mathcal{P}(n)}(C_1,C_2,\ldots,C_P;\mathbf{b})=\left[b_1 C_1 \mathbin\Vert \mathbf{0}^{n_1}  \mathbin\Vert b_2 C_2  \mathbin\Vert\mathbf{0}^{n_2}  \mathbin\Vert\cdots  \mathbin\Vert\mathbf{0}^{n_{P-1}}  \mathbin\Vert b_P C_P  \right],
\end{equation}
where, $\mathcal{P}(n)=(n_1,n_2,\ldots,n_{P-1})$ is partition of $n$ with $P-1$ non-negative integers, $n=n_1+n_2+\cdots+n_{P-1}$, $\mathbf{0}^{n_1}$ represents a zero matrix of size $M\times n_1$ and $ \mathbin\Vert$ represents concatenation of two matrices.
\end{construction}
First we consider a $(M,L)$-CCC, $P\le M$, such that $P\mid M$ and MOSs of length $P$. Now, for any positive integers $n$, we take a partition with $P-1$ non-negative integers. The partition of $n$ decides the position and numbers of nulls in the proposed codes.
\begin{theorem}\label{Th:SNC_CCC}
    Let $\mathbf{C}$ be a $(M,L)$-CCC, $P\mid M$, $\mathbf{b}^1,\mathbf{b}^2,\ldots,\mathbf{b}^P$ be MOSs of length $P$.  Now, define
    \begin{equation}
        B_{\nu P+\mu}=\mathcal{R}^{\mathcal{P}(n)}(C_{\nu P+1},C_{\nu P+2},\ldots,C_{(\nu+1)P};\mathbf{b}^\mu),
    \end{equation}
    for $0 \le \nu <\frac{M}{P},~ 1\le \mu \le P$. Then $\mathbf{B}=\{B_1,B_2,\ldots,B_{M}\}$ is a $(M,PL+n)$ SNC-CCC.
\end{theorem}
\begin{proof}
    Let $k_1$ and $k_2$ be two integers such that $k_1=\nu_1P+\mu_1$ and $k_2=\nu_2P+\mu_2$.
    Without loss of generality, we consider a random $C_{\nu_1P+i_1}$. For a corresponding $\tau$ in $\mathcal{C}(B_{k_1},B_{k_2})(\tau)$, $C_{\nu_1P+i_1}$ is correlated with matrix $X_{j_1}^{t_1,t_2}$, which is made of almost three parts, first $t_1$ columns of $X_{j_1}^{t_1,t_2}$ are $t_1$ last columns of $C_{\nu_2P+j_1}$, $n_{j_1}$ columns of $X_{j_1}^{t_1,t_2}$ are zero columns, and $t_2$ last columns of $X_{j_1}^{t_1,t_2}$ are $t_2$ first columns of $C_{\nu_2P+j_1+1}$, where $t_1,t_2\ge 0$ and $t_1+t_2+j_1=L$, in Fig. \ref{fig:relation_X_C} the $X_{j_1}^{t_1,t_2}$ is depicted.
    \begin{figure}[ht!]
        \centering
        \includegraphics[width=0.85\linewidth]{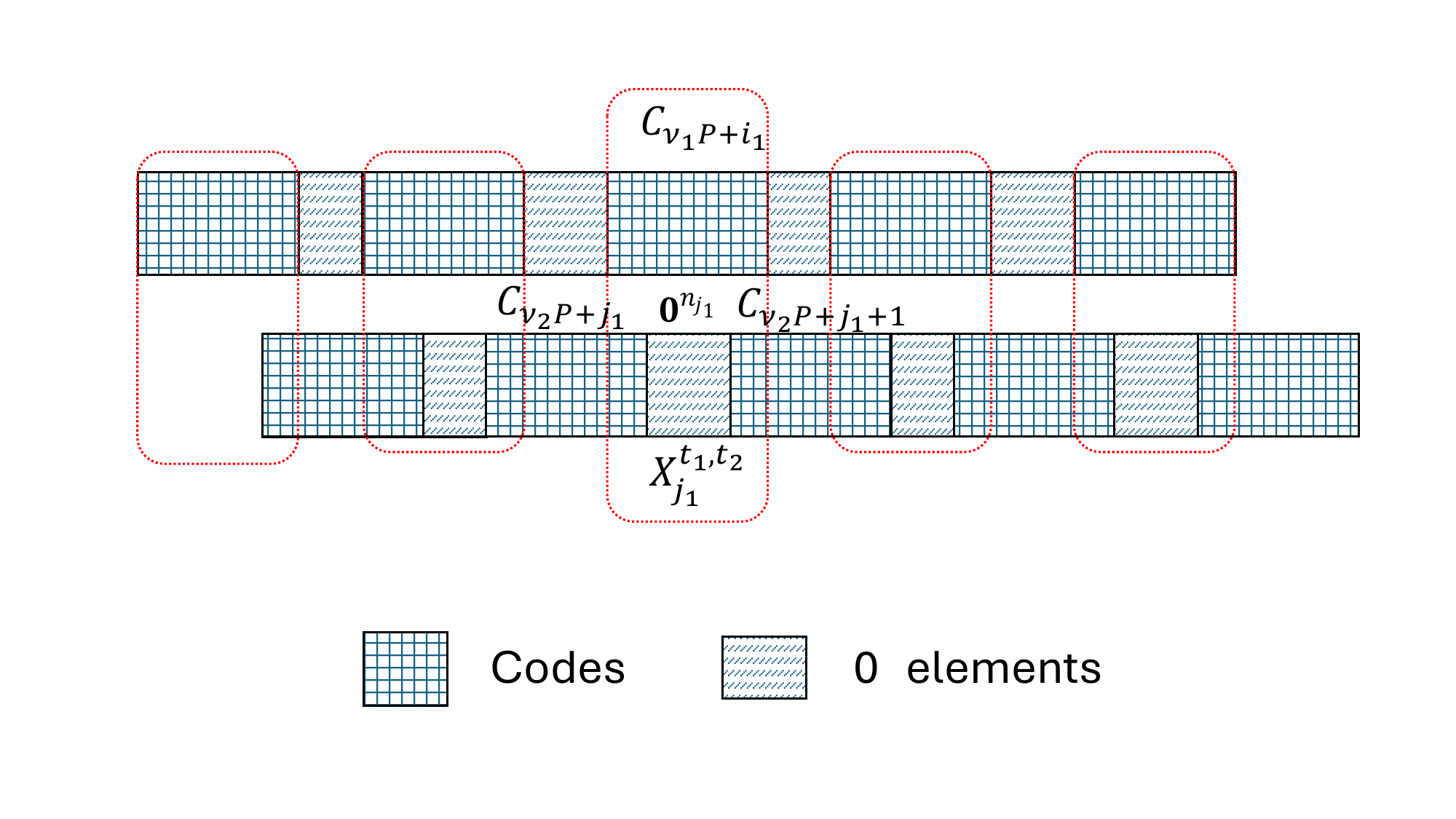}
        \caption{Illustration of $X_{j_1}^{t_1,t_2}$ and its correlation with $C_{\nu_1P+i_1}$}
        \label{fig:relation_X_C}
    \end{figure}
    
    The term $\mathcal{C}(C_{\nu_1P+i_1},X_{j_1}^{t_1,t_2})(0)$ has contribution  in $\mathcal{C}(B_{k_1},B_{k_2})(\tau)$, which is also depicted in Fig. \ref{fig:relation_X_C}. If we prove that $\mathcal{C}(C_{\nu_1P+i_1},X_{j_1}^{t_1,t_2})(0)=0$ for all $j_1,t_1,t_2$ then $\mathcal{C}(B_{k_1},B_{k_2})(\tau)=0$. For $\tau\ne 0$, we have
    \begin{equation}\label{eq_corr_c_x}
        \mathcal{C}(C_{\nu_1P+i_1},X_{j_1}^{t_1,t_2})(0)=\mathcal{C}(C_{\nu_1P+i_1},C_{\nu_2P+j_1})(j_1+t_3)+\mathcal{C}(C_{\nu_1P+i_1},C_{\nu_2P+j_2})(-(t_1+j_1)).
    \end{equation}
    Since $C_{\nu_1P+i_1}, C_{\nu_2P+j_1}, C_{\nu_2P+j_2} $ are GCSs of a CCC, it implies the right-hand side of eq. \eqref{eq_corr_c_x} is zero. Therefore, $\mathcal{C}(B_{k_1},B_{k_2})(\tau)=0$, for $\tau\ne 0$, and similarly $\mathcal{C}(B_{k_1})(\tau)=0$, for $\tau\ne 0$. Now, we check for $\mathcal{C}(B_{k_1},B_{k_2})(0)$ and $\mathcal{C}(B_{k_1})(0)$. Let $k_1\ne k_2$, we have
    \begin{equation}\label{Eq_corr_Bk1_Bk2}
        \mathcal{C}(B_{k_1},B_{k_2})(0)=\sum_{i=1}^{P}b^{\mu_1}_ib^{\mu_2*}_i \mathcal{C}(C_{\nu_1P+i},C_{\nu_2P+i})(0).
    \end{equation}
    In \eqref{Eq_corr_Bk1_Bk2}, either $\mu_1\ne \mu_2$ or $\nu_1\ne \nu_2$, which implies $\mathcal{C}(B_{k_1},B_{k_2})(0)=0$. When $k_1=k_2$, we have 
    \begin{equation}
        \mathcal{C}(B_{k_1})(0)=\sum_{i=1}^{P}b^{\mu_1}_ib^{\mu_1*}_i \mathcal{C}(C_{\nu_1P+i})(0)=PML.
    \end{equation}
    Therefore, $\mathbf{B}=\{B_1,B_2,\ldots,B_M\}$ is a $(M,PL+n)$ SNC-CCC.
\end{proof}
Since the alphabets of CCC and MOSs are identical, the resulting SNC-CCCs likewise possess the same alphabets, with an additional zero. The position of nulls can be determined by the partition of $n$, which is employed in the construction.
\begin{example}\label{Ex:SNC}
   Let 
    \begin{equation*}
        \begin{aligned}
            C_1=& \begin{bmatrix}
                + 	 + 	 +\\
                 + 	 + 	-\\
                - 	 + 	-\\
                - 	 + 	-
            \end{bmatrix} ,  C_2 =  \begin{bmatrix}
                 + 	- 	 +\\
                 + 	 + 	-\\
                - 	- 	 +\\
                 + 	 + 	 +
            \end{bmatrix},
            C_3=  \begin{bmatrix}
                 + 	- 	-\\
                 + 	 + 	 +\\
                - 	 +	 -\\
                 +	 -	 -
            \end{bmatrix} \text{and }  C_4 =
            \begin{bmatrix}
                 + 	- 	-\\
                 +	 -  +\\
                 + 	 +	  +\\
                - 	 + 	 +
            \end{bmatrix},
        \end{aligned}
    \end{equation*}
    be a $(4,3)$-CCC from \cite{Tao_SPL_ZCCS}, where $+$ and $-$ represent $1$ and $-1$, respectively. Now, assume $P=2$, $\mathbf{b}^1=(1,1)$, $\mathbf{b}^2=(1,-1)$, $n_1=3$. Then by \textbf{Theorem} \ref{Th:SNC_CCC}, $(B_1,B_2,B_3,B_4)$, such that,
\begin{equation}
    \begin{aligned}
        B_1&=   \begin{bmatrix}
             +  	 +  	 +  \,	0 \, 	0 \, 	0  	 +  	 -  	 +\\
 +  	 +  	 -  \,	0 \, 	0 \, 	0  	 +  	 +  	 -\\
 +  	 +  	 -  \,	0 \, 	0 \, 	0  	 -  	 -  	 +\\
 -  	 +  	 -  \,	0 \, 	0 \, 	0   	 +  	 +  	 +
        \end{bmatrix},
        &B_2=  \begin{bmatrix}
             +  	 +  	 + \,  0 \, 	0 \, 	0    -  	 +  	 -\\
 +  	 +	   -  \, 0 \, 	0 \, 	0     -  	 -  	 +\\
 +  	 +  	 - \, 0 \, 	0 \, 	0     +  	 +	   -\\
 -	   +	   - \, 0 \, 	0 \, 	0     -  	 -  	 -
        \end{bmatrix},\\
        B_3&=  \begin{bmatrix}
             +  	 -  	 - \,  0 \, 	0 \, 	0    +  	 -  	 -\\
 +  	 +  	 + \,  0 \, 	0 \, 	0     +	   -  	 +\\
 -  	 +	   -  \, 0 \, 	0 \, 	0    +  	 +	   +\\
 +	   -	   - \, 0 \, 	0 \, 	0    -  	 +  	 +
        \end{bmatrix},\text{ and }
        &B_4=  \begin{bmatrix}
             +  	 -	   - \, 0 \, 	0 \, 	0    -  	 +  	 +\\
 +  	 +  	 + \,  0 \, 	0 \, 	0    -  	 +	   -\\
 -	   +	   - \, 0 \, 	0 \, 	0     -  	 -  	 -\\
 +  	 -  	 - \, 0 \, 	0 \, 	0    +  	 -  	 -
        \end{bmatrix},
    \end{aligned}
\end{equation}
     is a $(4,9)$ SNC-CCC.
\end{example}
\begin{remark}
    In the \textbf{Example} \ref{Ex:SNC}, there is a freedom to decide on various values $n_1$.
\end{remark}
Now, we define a $P$ permutations $\pi_0,\pi_1,\ldots, \pi_{P-1}$ of the set $\{1,2,\ldots,M\}$, which will be used in our proposed construction. These permutations can be constructed as follows:

First consider $\pi_0$ to 
 be a random permutation of the set $\{1,2,\ldots,M\}$. Then we obtained the $i$-th permutation $\pi_i$ by cyclically shifting\footnote{The nature of the cyclic shift to be either left or right.} $\pi_1$ for $i$ times, where $1\le i \le P-1$. It can be observed that the properties of the permutations are as follows: 
\begin{equation}\label{Eq:def_pi}
    \pi_{j_1}(i_1P+\mu)\ne \pi_{j_2}(i_2P+\mu),
\end{equation}
for $j_1\ne j_2$, $0 \le i_1,i_2 < \frac{M}{P}$ and $1\le \mu \le P$.

\begin{theorem}\label{Th:SNC_MCCC}
     Let $\mathbf{C}$ be a $(M,L)$-CCC, $P\mid M$, $\mathbf{b}^1,\mathbf{b}^2,\ldots,\mathbf{b}^P$ be MOSs of length $P$ and $\pi_0,\pi_1,\ldots,\pi_{P-1}$ be permutations as defined in \eqref{Eq:def_pi}. Now, define
    \begin{equation}
        B_{jM+\nu P+\mu}=\mathcal{R}^{\mathcal{P}(n)}(C_{\pi_j(\nu P+1)},C_{\pi_j(\nu P+2)},\ldots,C_{\pi_j((\nu+1)P)};\mathbf{b}^\mu),
    \end{equation}
    for $0 \le \nu <\frac{M}{P},~ 1\le \mu \le P$ and $0\le j \le P-1$. Then, each $\mathbf{B}^j=\{B_{jM+1},B_{jM+2},\ldots,$ $B_{(j+1)M}\}$ is a $(M,PL+n)$ SNC-CCC and by combining all SNC-CCCs it becomes a multiple CCCs with inter-set ZCCZ.
\end{theorem}
\begin{proof}
    Let $B_{k_1}$ and $B_{k_2}$ be codes from two different CCCs. Using a similar argument can be made as in the proof of \textbf{Theorem} \ref{Th:SNC_CCC}, a $C_{\nu_1P+i_1}$ will be correlated with $X^{t_1,t_2}_{j_1}$. The term $\mathcal{C}(C_{\pi_{k_1}(\nu_1P+i_1)},X_{j_1}^{t_1,t_2})(0)$ is non-zero only when either $t_2=L$, $C_{\pi_{k_1}(\nu_1P+i_1)}=C_{\pi_{k_2}(\nu_2P+j_1+1)}$ or $t_1=L$, $C_{\pi_{k_1}(\nu_1P+i_1)}=C_{\pi_{k_2}(\nu_2P+j_1)}$. This condition will occur for a minimum time shift $\tau=l+\lambda$ for $\lambda=\min \{n_i:1 \le i < P\}$ and the maximum time shift $\tau=PL+n-L$. Thus, aperiodic ZCCZ width is $L+\lambda$, and periodic ZCCZ width is $L$.
\end{proof}
\begin{remark}
    The elements in the partition $n$ can be used to obtain the ZCCZ width of multiple SNC-CCCs. $L$ will be periodic and $\lambda+L$ be the aperiodic ZCCZ, where
    \begin{equation}
        \lambda=\min_{1\le i \le P-1}  \{n_i\}.
    \end{equation}
\end{remark}
Now, we present a method to decrease the magnitude of the cross-correlation value outside the ZCCZ.
\begin{corollary}\label{Cor:SNC_QCCS_n}
    Let $\mathcal{P}(n)$ be partition of $n$ i.e., $n=n_1+n_2+\ldots+n_{P-1}$ such that $n_{i_1}\ne n_{i_2} $ for $1\le i_1,i_2\le P-1$ in \textbf{Theorem} \ref{Th:SNC_MCCC} then $\delta_A$ becomes $LM$. Therefore, we have a aperiodic $(P,M,PL+n,L+\lambda,LM)$-CCCs. Further, if $n_{i_1}\ne n_{i_2} \mod{L}$ for $1\le i_1,i_2\le P-1$ in \textbf{Theorem} \ref{Th:SNC_MCCC} then $\delta_P$ becomes $LM$. This type of $\mathcal{P}(n)$ is possible for $n\ge \frac{P(P-1)}{2}$. Therefore, we have a periodic $(P,M,PL+n,L,LM)$-CCCs.
\end{corollary}
Let $\pi_1,\pi_2,\ldots,\pi_P$ be permutations of $\{1,2,\ldots,M\}$ as define in \eqref{Eq:def_pi}. Further, assume for any two permutation $\pi_{j_1}$ and $\pi_{j_2}$, when 
    \begin{equation}\label{Eq:pi_unique}
    \begin{aligned}
         \pi_{j_1}(i_1P+\mu_1)=& \pi_{j_2}(i_2P+\mu_2), \text{as well as}\\
         \pi_{j_1}(i_1P+\mu_1+\alpha) \ne & \pi_{j_2}(i_2P+\mu_2+\alpha),
    \end{aligned}
    \end{equation}
    for $1\le \mu_1+\alpha,\mu_2+\alpha \le P$.

When $M=P$, the set of permutations defined by \eqref{Eq:pi_unique} only, excluding the condition in \eqref{Eq:def_pi}, is defined as Tuscan square in \cite{chu2006tuscan}. From \cite{chu2006tuscan}, it can be verified that for the prime values of $M$, the total number of permutations can not exceed $M-1$, and the same holds for $M=8$ and $9$. Therefore, we can generate a set of permutations following \eqref{Eq:def_pi} and \eqref{Eq:pi_unique} not when $M=P$ is a prime and $8$ and $9$. We provide all possible permutations for $M$ up to $12$ with different values of $P$ in Table  \ref{tab:exam_permuation}; here, $A,B,$ and $ C$ represent $10,11,$ and $12$, respectively.
\begin{table}[!ht]
    \centering
    \tiny
    \begin{tabular}{|c|c|c|}
    \hline
        $M$ & $P$ & set of permutations \\
    \hline    
        $4$ & $2$ & $(3,2,1,4)$,$(4,3,2,1)$\\
        \hline
       $4$ & $4$ &  $(1,2,3,4),(4,3,2,1),(3,1,4,2),(2,4,1,3)$\\
       \hline 
       $6$ & $2$ & $(1,2,3,4,5,6),(6,1,4,3,2,5)$\\
       \hline
       $6$ & $3$ & $(1,2,3,4,5,6),(2,1,4,6,3,5),(5,4,2,3,6,1)$\\
       \hline
       $6$ & $6$ & $(3,1,2,6,5,4),(1,6,4,3,2,5),(2,4,1,5,3,6),(6,3,5,1,4,2),(5,2,3,4,6,1),(4,5,6,2,1,3)$\\
       \hline
       $8$ & $2$ & $(7,2,3,6,5,4,1,8),(2,1,4,7,8,3,6,5)$\\
       \hline
       $8$ & $4$ & $(3,2,1,4,8,5,7,6),(1,3,4,5,7,8,6,2),(2,4,3,7,5,6,8,1),(6,7,2,3,4,1,5,8)$\\
       \hline
       $9$ & $3$ & $(1,2,3,4,5,6,7,8,9),(6,1,8,9,4,2,3,7,5),(5,9,1,8,3,4,2,6,7)$  \\
       \hline
       $10$ & $2$ &$(1,2,3,4,5,6,7,8,9,A)(4,7,2,9,6,1,A,3,8,5)$ \\
       \hline
       $10$ & $5$ &$(1,2,3,4,5,6,7,8,9,A),(9,1,7,5,3,4,6,2,A,8),(5,4,1,8,2,A,9,6,3,7)$, \\
       & & $(2,8,5,1,4,7,3,A,6,9),(3,5,9,2,1,8,A,4,7,6)$\\
       \hline
        &  & $(1,5,2,4,8,6,A,7,3,9),(6,1,A,5,7,2,3,4,9,8),(4,7,1,8,3,5,6,9,2,A),$\\
       $10$ & $10$ & $(2,6,3,1,4,A,9,5,8,7),(3,A,8,2,1,9,7,6,4,5),(5,4,6,7,9,1,2,8,A,3),$\\
        & & $(7,8,5,9,A,4,1,3,6,2),(A,2,9,6,5,3,8,1,7,4),(8,9,4,3,2,7,5,A,1,6),$\\
        & &$(9,3,7,A,6,8,4,2,5,1)$\\
       \hline
       $12$ & $2$ & $(1,2,3,4,5,6,7,8,9,A,B,C)(6,1,4,3,2,5,C,7,A,9,8,B)$\\
       \hline
       $12$ & $3$ & $(1,2,3,4,5,6,7,8,9,A,B,C),(2,1,4,6,3,5,8,7,A,C,9,B),(3,4,2,5,C,1,9,A,8,B,6,7)$\\
       \hline
       $12$ & $4$ & $(1,2,3,4,5,6,7,8,9,A,B,C),(4,3,2,1,8,7,6,5,C,B,A,9),$\\
              & & $(3,1,4,2,B,9,C,A,7,5,8,6),(A,C,9,B,6,8,5,7,2,4,1,3)$ \\
       \hline
       $12$ & $6$ & $(3,7,2,6,5,4,9,1,8,C,B,A),(7,C,A,9,8,B,1,6,4,3,2,5), (8,4,1,5,3,6,2,A,7,B,9,C),$\\
       & & $(6,3,B,1,4,2,C,9,5,7,A,8), (B,2,9,A,C,7,5,8,3,4,6,1), (4,5,C,2,1,3,A,B,6,8,7,9)$ \\
\hline
       & & $(1,9,C,B,6,4,8,A,5,7,2,3),(8,1,A,9,5,C,7,B,2,6,3,4),(5,6,1,7,4,9,2,8,C,3,A,B)$\\
       $12$& $12$ & $(7,8,B,1,2,A,6,9,3,5,4,C),(A,C,2,4,1,5,B,3,8,9,7,6),(2,5,8,6,C,1,3,7,A,4,B,9)$\\
       & & $(9,B,4,A,7,3,1,C,6,8,5,2),(6,7,9,8,3,B,5,1,4,2,C,A),(C,4,5,3,9,6,A,2,1,B,8,7)$\\
       & & $(B,A,3,C,8,2,9,4,7,1,6,5),(4,3,6,2,B,7,C,5,9,A,1,8),(3,2,7,5,A,8,4,6,B,C,9,1)$\\
       \hline
    \end{tabular}
    \caption{Set of permutations that satisfy \eqref{Eq:def_pi} and \eqref{Eq:pi_unique} for $M$ up to $12$.}
    \label{tab:exam_permuation}
\end{table}

Further, We are providing a method to construct a set of permutations that satisfy \eqref{Eq:pi_unique}.

\begin{algorithm}
\caption{Algorithm to construct set of permutations that follow \eqref{Eq:def_pi} and \eqref{Eq:pi_unique}}\label{Algo_Int}
  \begin{algorithmic}[1]
\Require Input $P,M$ such that $P\mid M$
\Ensure $P$ many permutations following \eqref{Eq:pi_unique}
\State Generate $P$ random permutation of $\{1,2,\ldots,M\}$  $\pi_1,\pi_2,\ldots,\pi_P$.
\State Now, we refine each $\pi$ to follow \eqref{Eq:def_pi}
\For {$k=2:P$ \& $kk=1:k-1$ \& $i_1=0:m/P-1$ \& $i_2=0:m/P-1$ \& $\mu=1:P$}
\If {$\pi_{k}(i_1P+\mu)=\pi_{kk}(i_2P+\mu)$}
\State Replace $\pi_{kk}(i_1P+\mu)$ with element not in the set $\{ \pi_k(i_2P+\mu): \text{ for } kk=1\text{ to }k-1, i_2=0\text{ to }m/P-1 \}$.
\EndIf
\EndFor
\State Now, we define that follows \eqref{Eq:pi_unique}
\For {$k=2:P$ \& $kk=1:k-1$ \& $\mu_1 = 1:(P - 1)$ \& $\mu_2 = 1:(P - 1)$ \& $\alpha = 1:(P - \max\{\mu_1, \mu_2\})$ \& $i_1 = 0:(M/P - 1)$ \& $i_2 = 0:(M/P - 1)$}
\If{$\pi_k(i_1  P + \mu_1)=\pi_{kk}((i_2  P + \mu_2))$ \& $\pi_k(i_1  P + \mu_1+\alpha)=\pi_{kk}((i_2  P + \mu_2+\alpha))$}
\State Replace $\pi_k(i_1  P + \mu_1+\alpha)$ by elements of $\{\pi_k(i_3  P + \mu_1+\alpha):i_1\ne i_3\}$ such that $\{\pi_k(i_3  P + \mu_1+\alpha):i_1\ne i_3\}$ must not be an element of set $\{\pi_{kk}((i_2  P + \mu_2+\alpha)): \text{ for } kk=1\text{ to }k-1, i_2=0\text{ to }m/P-1 \}$
\EndIf
\EndFor

\Return $\pi_1,\pi_2,\ldots,\pi_P$
\end{algorithmic}
\end{algorithm}
\begin{corollary}\label{Cor:SNC_QCCS_pi}
    \textbf{Theorem} \ref{Th:SNC_MCCC} provides multiple SNC-CCCs such that cross-correlation is non-zero at only one time shift and the value is equal to $LM$ when permutations satisfy \eqref{Eq:pi_unique}. Therefore, we have both aperiodic $(P,M,PL+n,L+\lambda,LM)$-CCCs and periodic $(P,M,PL+n,L,LM)$-CCCs.
\end{corollary}
We are providing one more example to illustrate the multiple SNC-CCCs with ZCCZ.
\begin{example}
In  \textbf{Example} \ref{Ex:SNC}, we considered as identity permutation ($\pi_1=(1,2,3,4)$) and construct a SNC-CCC. As $P=2$, we have second permutation $\pi_2=(4,1,2,3)$ such that $\pi_1$ and $\pi_2$ satisfy \eqref{Eq:pi_unique}. With the same $\mathbf{C}$, $\mathbf{b}^1$ and $\mathbf{b}^2$, we have one CCC as given in \textbf{Example} \ref{Ex:SNC} and the other CCC is as given below
 \begin{equation}
    \begin{aligned}
        B_5&=   \begin{bmatrix}
            +-- \,0\,0\,0 +++\\
            +-+ \,0\,0\,0 ++-\\
            +++ \,0\,0\,0 -+-\\
            -++ \,0\,0\,0 -+-
        \end{bmatrix},
        &B_6=  \begin{bmatrix}
            +-- \,0\,0\,0 ---\\
            +-+ \,0\,0\,0 --+\\
            +++ \,0\,0\,0 +-+\\
            -++ \,0\,0\,0 +-+
        \end{bmatrix},\\
        B_7&=  \begin{bmatrix}
            +-+ \,0\,0\,0 +--\\
            ++- \,0\,0\,0 +++\\
            --+ \,0\,0\,0 -+-\\
            +++ \,0\,0\,0 +--
        \end{bmatrix},\text{ and }
        &B_8=  \begin{bmatrix}
            +-+ \,0\,0\,0 -++\\
            ++- \,0\,0\,0 ---\\
            --+ \,0\,0\,0 +-+\\
            +++ \,0\,0\,0 -++
        \end{bmatrix}.
    \end{aligned}
\end{equation}
Combining $B_1,B_2,\ldots,B_8$ is a multiple SNC-CCCs with aperiodic ZCCZ width is $6$.
\end{example}

Choosing a set of permutations plays a role in a low correlation magnitude value. To ensure such permutation exists, we are providing an example. Let $M=4$, $P=4$ so that $\pi_1=(1,2,3,4),\;\pi_2=(2,3,4,1),\;\pi_3=(3,4,1,2)$ and $\pi_4=(4,1,2,3)$, satisfy \eqref{Eq:def_pi} but not \eqref{Eq:pi_unique}. But $\pi_1=(1,2,3,4),\;\pi_2=(4,3,2,1),\;\pi_3=(3,1,4,2)$ and $\pi_4=(2,4,1,3)$ satisfies \eqref{Eq:pi_unique}.


\section{Comparison}\label{Sec:comparison}
Here, we compare our construction, highlighting its novelty and advantages compared to the existing literature.
\subsection{Comparison with \cite{Sahin_2021} and \cite{Shen_2022}}
 \cite{Sahin_2021} and \cite{Shen_2022} generated SNC-GCS via generalised Boolean functions and generated parameters closely multiple of $2$. Moreover, their findings are restricted to SNC-GCSs and do not extend to SNC-CCCs. However, in the proposed construction, every code is a SNC-GCS, which is not restricted to a multiple of $2$ only.
 \subsection{Comparision with \cite{Zhou_2020}}
 The method in \cite{Zhou_2020} implements an iterative strategy using an existing CCC to generate SNC-MOGCSs. Whereas, the proposed construction provides SNC-CCCs.
 \subsection{Comparision with \cite{Shen_CCC}}
SNC-CCCs were constructed using an extended Boolean function in  \cite{Shen_CCC} with diversified parameters. Notably, these parameters consistently involved powers of $p$ for $p\ge 2$, where the elements are obtained from the $q$th root of unity and zero for $(p \mid q)$. The given example \ref{Ex:SNC} demonstrates that the proposed construction offers more flexibility.

\subsection{Comparison with \cite{Rajen_SZCCS}}
When $n=0$, the proposed construction reduces to the traditional CCC and multiple traditional CCC, which is given in \cite{Rajen_SZCCS}. Therefore, CCC and multiple CCCs of \cite{Rajen_SZCCS} are a special case of the proposed construction.
 \section{Future directions}\label{Sec:fut_work}
 Based on our contribution in this paper, we would like to introduce the following future works: 
 \begin{enumerate}
 \item In \textbf{Corollary} \ref{Cor:SNC_QCCS_pi}, the set of permutations used satisfying \eqref{Eq:def_pi} and \eqref{Eq:pi_unique}. Although we provide an algorithm to generate it and some examples, defining a function-based construction may still be interesting. We shall consider it to be our future research problem.
 
     \item In the current literature, we do not have sufficient information on the optimal collection of multiple SNC-CCCs in relation to their maximum magnitude of inter-set cross-correlation value. This limitation leads to a future direction on deriving a lower correlation bound for multiple collections of SNC-CCCs.
     \item Besides, as we can see in our proposed construction of multiple SNC-CCCs, we also have a ZCCZ, which leads us to the natural question, ``What will be the relationship between the ZCCZ width and the other parameters of our multiple collection of SNC-CCCs?''.
 \end{enumerate}

\section{Conclusion}\label{Sec:Con}
In this paper, with the help of MOSs, we developed a method to construct SNC-CCCs, with flexible parameters. The proposed construction can cover almost all the possible lengths over the alphabet $\{-1,0,1\}$. Further, we have extended the construction to produce multiple SNC-CCCs with inter-set ZCCZ. The proposed construction includes a wider range of parameters in relation to length and alphabet size. Furthermore, we have shown that restriction can be made on the highest cross-correlation magnitude value outside the ZCCZ width, assuring that the multiple SNC-CCCs possess not only inter-set ZCCZ width but also exhibit a low cross-correlation magnitude value outside the ZCCZ width.


\bibliography{sn-bibliography}


\begin{thebibliography}{36}
\ifx \bisbn   \undefined \def \bisbn  #1{ISBN #1}\fi
\ifx \binits  \undefined \def \binits#1{#1}\fi
\ifx \bauthor  \undefined \def \bauthor#1{#1}\fi
\ifx \batitle  \undefined \def \batitle#1{#1}\fi
\ifx \bjtitle  \undefined \def \bjtitle#1{#1}\fi
\ifx \bvolume  \undefined \def \bvolume#1{\textbf{#1}}\fi
\ifx \byear  \undefined \def \byear#1{#1}\fi
\ifx \bissue  \undefined \def \bissue#1{#1}\fi
\ifx \bfpage  \undefined \def \bfpage#1{#1}\fi
\ifx \blpage  \undefined \def \blpage #1{#1}\fi
\ifx \burl  \undefined \def \burl#1{\textsf{#1}}\fi
\ifx \doiurl  \undefined \def \doiurl#1{\url{https://doi.org/#1}}\fi
\ifx \betal  \undefined \def \betal{\textit{et al.}}\fi
\ifx \binstitute  \undefined \def \binstitute#1{#1}\fi
\ifx \binstitutionaled  \undefined \def \binstitutionaled#1{#1}\fi
\ifx \bctitle  \undefined \def \bctitle#1{#1}\fi
\ifx \beditor  \undefined \def \beditor#1{#1}\fi
\ifx \bpublisher  \undefined \def \bpublisher#1{#1}\fi
\ifx \bbtitle  \undefined \def \bbtitle#1{#1}\fi
\ifx \bedition  \undefined \def \bedition#1{#1}\fi
\ifx \bseriesno  \undefined \def \bseriesno#1{#1}\fi
\ifx \blocation  \undefined \def \blocation#1{#1}\fi
\ifx \bsertitle  \undefined \def \bsertitle#1{#1}\fi
\ifx \bsnm \undefined \def \bsnm#1{#1}\fi
\ifx \bsuffix \undefined \def \bsuffix#1{#1}\fi
\ifx \bparticle \undefined \def \bparticle#1{#1}\fi
\ifx \barticle \undefined \def \barticle#1{#1}\fi
\bibcommenthead
\ifx \bconfdate \undefined \def \bconfdate #1{#1}\fi
\ifx \botherref \undefined \def \botherref #1{#1}\fi
\ifx \url \undefined \def \url#1{\textsf{#1}}\fi
\ifx \bchapter \undefined \def \bchapter#1{#1}\fi
\ifx \bbook \undefined \def \bbook#1{#1}\fi
\ifx \bcomment \undefined \def \bcomment#1{#1}\fi
\ifx \oauthor \undefined \def \oauthor#1{#1}\fi
\ifx \citeauthoryear \undefined \def \citeauthoryear#1{#1}\fi
\ifx \endbibitem  \undefined \def \endbibitem {}\fi
\ifx \bconflocation  \undefined \def \bconflocation#1{#1}\fi
\ifx \arxivurl  \undefined \def \arxivurl#1{\textsf{#1}}\fi
\csname PreBibitemsHook\endcsname

\bibitem[\protect\citeauthoryear{Golay}{1951}]{Golay:51}
\begin{barticle}
\bauthor{\bsnm{Golay}, \binits{M.J.E.}}:
\batitle{Static multislit spectrometry and its application to the panoramic display of infrared spectra$\ast$}.
\bjtitle{J. Opt. Soc. Am.}
\bvolume{41}(\bissue{7}),
\bfpage{468}--\blpage{472}
(\byear{1951})
\doiurl{10.1364/JOSA.41.000468}
\end{barticle}
\endbibitem

\bibitem[\protect\citeauthoryear{Golay}{1961}]{Golay_1961}
\begin{barticle}
\bauthor{\bsnm{Golay}, \binits{M.}}:
\batitle{Complementary series}.
\bjtitle{IRE Transactions on Information Theory}
\bvolume{7}(\bissue{2}),
\bfpage{82}--\blpage{87}
(\byear{1961})
\doiurl{10.1109/TIT.1961.1057620}
\end{barticle}
\endbibitem

\bibitem[\protect\citeauthoryear{Spasojevic and Georghiades}{2001}]{Spasojevic}
\begin{barticle}
\bauthor{\bsnm{Spasojevic}, \binits{P.}},
\bauthor{\bsnm{Georghiades}, \binits{C.N.}}:
\batitle{Complementary sequences for {ISI} channel estimation}.
\bjtitle{IEEE Trans. Inf. Theory}
\bvolume{47}(\bissue{3}),
\bfpage{1145}--\blpage{1152}
(\byear{2001})
\doiurl{10.1109/18.915670}
\end{barticle}
\endbibitem

\bibitem[\protect\citeauthoryear{Wang et~al.}{2007}]{Wang_IET}
\begin{barticle}
\bauthor{\bsnm{Wang}, \binits{H.M.}},
\bauthor{\bsnm{Gao}, \binits{X.Q.}},
\bauthor{\bsnm{Jiang}, \binits{B.}},
\bauthor{\bsnm{You}, \binits{X.H.}},
\bauthor{\bsnm{Hong}, \binits{W.}}:
\batitle{Efficient {MIMO} channel estimation using complementary sequences}.
\bjtitle{IET Communications}
\bvolume{1},
\bfpage{962}--\blpage{9697}
(\byear{2007})
\end{barticle}
\endbibitem

\bibitem[\protect\citeauthoryear{Șahin and Yang}{2020}]{Sahin}
\begin{barticle}
\bauthor{\bsnm{Șahin}, \binits{A.}},
\bauthor{\bsnm{Yang}, \binits{R.}}:
\batitle{An uplink control channel design with complementary sequences for unlicensed bands}.
\bjtitle{IEEE Trans. Wireless Commun.}
\bvolume{19}(\bissue{10}),
\bfpage{6858}--\blpage{6870}
(\byear{2020})
\doiurl{10.1109/TWC.2020.3006395}
\end{barticle}
\endbibitem

\bibitem[\protect\citeauthoryear{Yu}{2021}]{Yu}
\begin{barticle}
\bauthor{\bsnm{Yu}, \binits{N.Y.}}:
\batitle{Binary {G}olay spreading sequences and {R}eed-{M}uller codes for uplink grant-free {NOMA}}.
\bjtitle{IEEE Trans. Commun.}
\bvolume{69}(\bissue{1}),
\bfpage{276}--\blpage{290}
(\byear{2021})
\doiurl{10.1109/TCOMM.2020.3031613}
\end{barticle}
\endbibitem

\bibitem[\protect\citeauthoryear{Pezeshki et~al.}{2008}]{Pezeshki}
\begin{barticle}
\bauthor{\bsnm{Pezeshki}, \binits{A.}},
\bauthor{\bsnm{Calderbank}, \binits{A.R.}},
\bauthor{\bsnm{Moran}, \binits{W.}},
\bauthor{\bsnm{Howard}, \binits{S.D.}}:
\batitle{Doppler resilient {G}olay complementary waveforms}.
\bjtitle{IEEE Trans. Inf. Theory}
\bvolume{54}(\bissue{9}),
\bfpage{4254}--\blpage{4266}
(\byear{2008})
\doiurl{10.1109/TIT.2008.928292}
\end{barticle}
\endbibitem

\bibitem[\protect\citeauthoryear{Davis and Jedwab}{1999}]{Davis}
\begin{barticle}
\bauthor{\bsnm{Davis}, \binits{J.A.}},
\bauthor{\bsnm{Jedwab}, \binits{J.}}:
\batitle{Peak-to-mean power control in {OFDM}, {G}olay complementary sequences, and {R}eed-{M}uller codes}.
\bjtitle{IEEE Trans. Inf. Theory}
\bvolume{45}(\bissue{7}),
\bfpage{2397}--\blpage{2417}
(\byear{1999})
\doiurl{10.1109/18.796380}
\end{barticle}
\endbibitem

\bibitem[\protect\citeauthoryear{Tseng and Liu}{1972}]{Tseng}
\begin{barticle}
\bauthor{\bsnm{Tseng}, \binits{C.-C.}},
\bauthor{\bsnm{Liu}, \binits{C.}}:
\batitle{Complementary sets of sequences}.
\bjtitle{IEEE Trans. Inf. Theory}
\bvolume{18}(\bissue{5}),
\bfpage{644}--\blpage{652}
(\byear{1972})
\doiurl{10.1109/TIT.1972.1054860}
\end{barticle}
\endbibitem

\bibitem[\protect\citeauthoryear{Nguyêñ and Coxson}{2016}]{Nguyen}
\begin{barticle}
\bauthor{\bsnm{Nguyêñ}, \binits{H.D.}},
\bauthor{\bsnm{Coxson}, \binits{G.E.}}:
\batitle{Doppler tolerance, complementary code sets, and generalised thue–morse sequences}.
\bjtitle{IET Radar, Sonar \& Amp; Navigation}
\bvolume{10},
\bfpage{1603}--\blpage{1610}
(\byear{2016})
\doiurl{10.1049/iet-rsn.2015.0569}
\end{barticle}
\endbibitem

\bibitem[\protect\citeauthoryear{Aparicio and Shimura}{2018}]{Aparicio}
\begin{barticle}
\bauthor{\bsnm{Aparicio}, \binits{J.}},
\bauthor{\bsnm{Shimura}, \binits{T.}}:
\batitle{Asynchronous detection and identification of multiple users by multi-carrier modulated complementary set of sequences}.
\bjtitle{IEEE Access}
\bvolume{6},
\bfpage{22054}--\blpage{22069}
(\byear{2018})
\doiurl{10.1109/ACCESS.2018.2828500}
\end{barticle}
\endbibitem

\bibitem[\protect\citeauthoryear{Paterson}{2000}]{Paterson}
\begin{barticle}
\bauthor{\bsnm{Paterson}, \binits{K.G.}}:
\batitle{Generalized {R}eed-{M}uller codes and power control in {OFDM} modulation}.
\bjtitle{IEEE Trans. Inf. Theory}
\bvolume{46}(\bissue{1}),
\bfpage{104}--\blpage{120}
(\byear{2000})
\doiurl{10.1109/18.817512}
\end{barticle}
\endbibitem

\bibitem[\protect\citeauthoryear{Schmidt}{2007}]{Schmidt}
\begin{barticle}
\bauthor{\bsnm{Schmidt}, \binits{K.-U.}}:
\batitle{Complementary sets, generalized {R}eed–{M}uller codes, and power control for {OFDM}}.
\bjtitle{IEEE Trans. Inf. Theory}
\bvolume{53}(\bissue{2}),
\bfpage{808}--\blpage{814}
(\byear{2007})
\doiurl{10.1109/TIT.2006.889723}
\end{barticle}
\endbibitem

\bibitem[\protect\citeauthoryear{Sarkar et~al.}{2020}]{Sarkar_2020}
\begin{barticle}
\bauthor{\bsnm{Sarkar}, \binits{P.}},
\bauthor{\bsnm{Majhi}, \binits{S.}},
\bauthor{\bsnm{Liu}, \binits{Z.}}:
\batitle{A direct and generalized construction of polyphase complementary sets with low {PMEPR} and high code-rate for {OFDM} system}.
\bjtitle{IEEE Trans. Commun.}
\bvolume{68}(\bissue{10}),
\bfpage{6245}--\blpage{6262}
(\byear{2020})
\doiurl{10.1109/TCOMM.2020.3007390}
\end{barticle}
\endbibitem

\bibitem[\protect\citeauthoryear{Suehiro and Hatori}{1988}]{Suehiro}
\begin{barticle}
\bauthor{\bsnm{Suehiro}, \binits{N.}},
\bauthor{\bsnm{Hatori}, \binits{M.}}:
\batitle{N-shift cross-orthogonal sequences}.
\bjtitle{IEEE Trans. Inf. Theory}
\bvolume{34}(\bissue{1}),
\bfpage{143}--\blpage{146}
(\byear{1988})
\doiurl{10.1109/18.2615}
\end{barticle}
\endbibitem

\bibitem[\protect\citeauthoryear{Liu et~al.}{2021}]{Liu_2021}
\begin{barticle}
\bauthor{\bsnm{Liu}, \binits{X.}},
\bauthor{\bsnm{Huang}, \binits{Y.}},
\bauthor{\bsnm{Chang}, \binits{C.-Y.}},
\bauthor{\bsnm{Chen}, \binits{H.-H.}}:
\batitle{Generalized complementary coded scrambling multiple access for {MIMO} communications}.
\bjtitle{IEEE Trans. Veh. Technol.}
\bvolume{70}(\bissue{12}),
\bfpage{13047}--\blpage{13061}
(\byear{2021})
\doiurl{10.1109/TVT.2021.3125029}
\end{barticle}
\endbibitem

\bibitem[\protect\citeauthoryear{Sun et~al.}{2015}]{Sun_2015_Survey}
\begin{barticle}
\bauthor{\bsnm{Sun}, \binits{S.-Y.}},
\bauthor{\bsnm{Chen}, \binits{H.-H.}},
\bauthor{\bsnm{Meng}, \binits{W.-X.}}:
\batitle{A survey on complementary-coded {MIMO} {CDMA} wireless communications}.
\bjtitle{IEEE Commun. Surv. Tutor.}
\bvolume{17}(\bissue{1}),
\bfpage{52}--\blpage{69}
(\byear{2015})
\doiurl{10.1109/COMST.2014.2332999}
\end{barticle}
\endbibitem

\bibitem[\protect\citeauthoryear{Tang et~al.}{2014}]{Tang_2014}
\begin{barticle}
\bauthor{\bsnm{Tang}, \binits{J.}},
\bauthor{\bsnm{Zhang}, \binits{N.}},
\bauthor{\bsnm{Ma}, \binits{Z.}},
\bauthor{\bsnm{Tang}, \binits{B.}}:
\batitle{Construction of doppler resilient complete complementary code in {MIMO} radar}.
\bjtitle{IEEE Trans. Signal Process.}
\bvolume{62}(\bissue{18}),
\bfpage{4704}--\blpage{4712}
(\byear{2014})
\doiurl{10.1109/TSP.2014.2337272}
\end{barticle}
\endbibitem

\bibitem[\protect\citeauthoryear{Sun et~al.}{2015}]{Sun_2015_VT}
\begin{barticle}
\bauthor{\bsnm{Sun}, \binits{S.-Y.}},
\bauthor{\bsnm{Chen}, \binits{H.-H.}},
\bauthor{\bsnm{Meng}, \binits{W.-X.}}:
\batitle{A framework to construct three-dimensional complementary codes for multiuser {MIMO} systems}.
\bjtitle{IEEE Trans. Veh. Technol.}
\bvolume{64}(\bissue{7}),
\bfpage{2861}--\blpage{2874}
(\byear{2015})
\doiurl{10.1109/TVT.2014.2351254}
\end{barticle}
\endbibitem

\bibitem[\protect\citeauthoryear{Men and Li}{2022}]{Men_MCCC}
\begin{barticle}
\bauthor{\bsnm{Men}, \binits{X.}},
\bauthor{\bsnm{Li}, \binits{Y.}}:
\batitle{New construction of multiple complete complementary codes with inter-set zero cross-correlation zone}.
\bjtitle{IEEE Signal Process. Lett.}
\bvolume{29},
\bfpage{1958}--\blpage{1962}
(\byear{2022})
\doiurl{10.1109/LSP.2022.3205273}
\end{barticle}
\endbibitem

\bibitem[\protect\citeauthoryear{Tian et~al.}{2020}]{Tian_MCCC}
\begin{barticle}
\bauthor{\bsnm{Tian}, \binits{L.}},
\bauthor{\bsnm{Li}, \binits{Y.}},
\bauthor{\bsnm{Xu}, \binits{C.}}:
\batitle{Multiple complete complementary codes with inter-set zero cross-correlation zone}.
\bjtitle{IEEE Trans. Commun.}
\bvolume{68}(\bissue{3}),
\bfpage{1925}--\blpage{1936}
(\byear{2020})
\doiurl{10.1109/TCOMM.2019.2962037}
\end{barticle}
\endbibitem

\bibitem[\protect\citeauthoryear{}{1997}]{IEEE_std}
\begin{botherref}
:
{IEEE} standard for wireless lan medium access control {(MAC)} and physical layer {(PHY)} specifications.
IEEE Std 802.11-1997,
1--445
(1997)
\doiurl{10.1109/IEEESTD.1997.85951}
\end{botherref}
\endbibitem

\bibitem[\protect\citeauthoryear{Zhou et~al.}{2020}]{Zhou_2020}
\begin{barticle}
\bauthor{\bsnm{Zhou}, \binits{Y.}},
\bauthor{\bsnm{Yang}, \binits{Y.}},
\bauthor{\bsnm{Zhou}, \binits{Z.}},
\bauthor{\bsnm{Anand}, \binits{K.}},
\bauthor{\bsnm{Hu}, \binits{S.}},
\bauthor{\bsnm{Guan}, \binits{Y.L.}}:
\batitle{New complementary sets with low {PAPR} property under spectral null constraints}.
\bjtitle{IEEE Trans. Inf. Theory}
\bvolume{66}(\bissue{11}),
\bfpage{7022}--\blpage{7032}
(\byear{2020})
\doiurl{10.1109/TIT.2020.3024984}
\end{barticle}
\endbibitem

\bibitem[\protect\citeauthoryear{Hamilton et~al.}{2011}]{Hamilton}
\begin{barticle}
\bauthor{\bsnm{Hamilton}, \binits{B.R.}},
\bauthor{\bsnm{Ma}, \binits{X.}},
\bauthor{\bsnm{Kleider}, \binits{J.E.}},
\bauthor{\bsnm{Baxley}, \binits{R.J.}}:
\batitle{{OFDM} pilot design for channel estimation with null edge subcarriers}.
\bjtitle{IEEE Trans. Wireless Commun.}
\bvolume{10}(\bissue{10}),
\bfpage{3145}--\blpage{3150}
(\byear{2011})
\doiurl{10.1109/TWC.2011.090611.101922}
\end{barticle}
\endbibitem

\bibitem[\protect\citeauthoryear{Su et~al.}{2019}]{Su_2019}
\begin{barticle}
\bauthor{\bsnm{Su}, \binits{D.}},
\bauthor{\bsnm{Jiang}, \binits{Y.}},
\bauthor{\bsnm{Wang}, \binits{X.}},
\bauthor{\bsnm{Gao}, \binits{X.}}:
\batitle{Omnidirectional precoding for massive {MIMO} with uniform rectangular array—part {I}: Complementary codes-based schemes}.
\bjtitle{IEEE Trans. Signal Process.}
\bvolume{67}(\bissue{18}),
\bfpage{4761}--\blpage{4771}
(\byear{2019})
\doiurl{10.1109/TSP.2019.2931205}
\end{barticle}
\endbibitem

\bibitem[\protect\citeauthoryear{Gavish and Lempel}{1994}]{Gavish_ternary}
\begin{barticle}
\bauthor{\bsnm{Gavish}, \binits{A.}},
\bauthor{\bsnm{Lempel}, \binits{A.}}:
\batitle{On ternary complementary sequences}.
\bjtitle{IEEE Trans. Inf. Theory}
\bvolume{40}(\bissue{2}),
\bfpage{522}--\blpage{526}
(\byear{1994})
\doiurl{10.1109/18.312179}
\end{barticle}
\endbibitem

\bibitem[\protect\citeauthoryear{Şahin and Yang}{2021}]{Sahin_2021}
\begin{barticle}
\bauthor{\bsnm{Şahin}, \binits{A.}},
\bauthor{\bsnm{Yang}, \binits{R.}}:
\batitle{A generic complementary sequence construction and associated encoder/decoder design}.
\bjtitle{IEEE Trans. Commun.}
\bvolume{69}(\bissue{11}),
\bfpage{7691}--\blpage{7705}
(\byear{2021})
\doiurl{10.1109/TCOMM.2021.3103554}
\end{barticle}
\endbibitem

\bibitem[\protect\citeauthoryear{Ipanov et~al.}{2018}]{Ipanov}
\begin{barticle}
\bauthor{\bsnm{Ipanov}, \binits{R.N.}},
\bauthor{\bsnm{Baskakov}, \binits{A.I.}},
\bauthor{\bsnm{Olyunin}, \binits{N.}},
\bauthor{\bsnm{Ka}, \binits{M.-H.}}:
\batitle{Radar signals with {ZACZ} based on pairs of {D}-code sequences and their compression algorithm}.
\bjtitle{IEEE Signal Process. Lett.}
\bvolume{25}(\bissue{10}),
\bfpage{1560}--\blpage{1564}
(\byear{2018})
\doiurl{10.1109/LSP.2018.2867734}
\end{barticle}
\endbibitem

\bibitem[\protect\citeauthoryear{Shen et~al.}{2023}]{Shen_CCC}
\begin{barticle}
\bauthor{\bsnm{Shen}, \binits{B.}},
\bauthor{\bsnm{Yang}, \binits{Y.}},
\bauthor{\bsnm{Zhou}, \binits{Z.}},
\bauthor{\bsnm{Mesnager}, \binits{S.}}:
\batitle{Constructions of spectrally null constrained complete complementary codes via the graph of extended boolean functions}.
\bjtitle{IEEE Trans. Inf. Theory}
\bvolume{69}(\bissue{9}),
\bfpage{6028}--\blpage{6039}
(\byear{2023})
\doiurl{10.1109/TIT.2023.3280210}
\end{barticle}
\endbibitem

\bibitem[\protect\citeauthoryear{Kumar et~al.}{2024}]{Nishant_SNC}
\begin{barticle}
\bauthor{\bsnm{Kumar}, \binits{N.}},
\bauthor{\bsnm{Sarkar}, \binits{P.}},
\bauthor{\bsnm{Majhi}, \binits{S.}}:
\batitle{Construction of spectrally-null-constrained zero-correlation zone sequences with flexible support}.
\bjtitle{Cryptography and Communications}
(\byear{2024})
\doiurl{10.1007/s12095-024-00715-0}
\end{barticle}
\endbibitem

\bibitem[\protect\citeauthoryear{Sarkar et~al.}{2024}]{Sarkar_QCCS_TIT}
\begin{barticle}
\bauthor{\bsnm{Sarkar}, \binits{P.}},
\bauthor{\bsnm{Li}, \binits{C.}},
\bauthor{\bsnm{Majhi}, \binits{S.}},
\bauthor{\bsnm{Liu}, \binits{Z.}}:
\batitle{New correlation bound and construction of {Q}uasi-complementary sequence sets}.
\bjtitle{IEEE Trans. Inf. Theory}
\bvolume{70}(\bissue{3}),
\bfpage{2201}--\blpage{2223}
(\byear{2024})
\doiurl{10.1109/TIT.2024.3352895}
\end{barticle}
\endbibitem

\bibitem[\protect\citeauthoryear{Kumar et~al.}{2023}]{Rajen_DZCCS}
\begin{botherref}
\oauthor{\bsnm{Kumar}, \binits{R.}},
\oauthor{\bsnm{Srivastava}, \binits{P.K.}},
\oauthor{\bsnm{Majhi}, \binits{S.}}:
A Direct Construction of Type-{II} {Z} Complementary Code Set with Arbitrarily Large Codes
(2023).
\url{https://arxiv.org/abs/2305.01290}
\end{botherref}
\endbibitem

\bibitem[\protect\citeauthoryear{Tao et~al.}{2022}]{Tao_SPL_ZCCS}
\begin{barticle}
\bauthor{\bsnm{Tao}, \binits{Y.}},
\bauthor{\bsnm{Avik}, \binits{A.} \bsuffix{Ranjan}},
\bauthor{\bsnm{Yanyan}, \binits{W.}},
\bauthor{\bsnm{Yang}, \binits{Y.}}:
\batitle{New class of optimal {Z}-complementary code sets}.
\bjtitle{IEEE Signal Process. Lett.}
\bvolume{29},
\bfpage{1477}--\blpage{1481}
(\byear{2022})
\doiurl{SPL-33621-2022}
\end{barticle}
\endbibitem

\bibitem[\protect\citeauthoryear{Chu et~al.}{2006}]{chu2006tuscan}
\begin{bchapter}
\bauthor{\bsnm{Chu}, \binits{W.}},
\bauthor{\bsnm{Golomb}, \binits{S.W.}},
\bauthor{\bsnm{Song}, \binits{H.-Y.}}:
\bctitle{Tuscan squares}.
In: \bbtitle{Handbook of Combinatorial Designs},
pp. \bfpage{678}--\blpage{682}.
\bpublisher{Chapman and Hall/CRC},
\blocation{Boca Raton, FL}
(\byear{2006})
\end{bchapter}
\endbibitem

\bibitem[\protect\citeauthoryear{Shen et~al.}{2022}]{Shen_2022}
\begin{barticle}
\bauthor{\bsnm{Shen}, \binits{B.}},
\bauthor{\bsnm{Yang}, \binits{Y.}},
\bauthor{\bsnm{Fan}, \binits{P.}},
\bauthor{\bsnm{Zhou}, \binits{Z.}}:
\batitle{Constructions of non-contiguous complementary sequence sets and their applications}.
\bjtitle{IEEE Trans. Wireless Commun.}
\bvolume{21}(\bissue{7}),
\bfpage{4871}--\blpage{4882}
(\byear{2022})
\doiurl{10.1109/TWC.2021.3133629}
\end{barticle}
\endbibitem

\bibitem[\protect\citeauthoryear{Kumar et~al.}{2024}]{Rajen_SZCCS}
\begin{bchapter}
\bauthor{\bsnm{Kumar}, \binits{R.}},
\bauthor{\bsnm{Srivastava}, \binits{P.K.}},
\bauthor{\bsnm{Majhi}, \binits{S.}}:
\bctitle{A new construction of optimal symmetrical {ZCCS}}.
In: \bbtitle{2024 IEEE International Symposium on Information Theory (ISIT)},
pp. \bfpage{1748}--\blpage{1752}
(\byear{2024}).
\doiurl{10.1109/ISIT57864.2024.10619512}
\end{bchapter}
\endbibitem

\end{thebibliography}

\end{document}